\documentclass[12pt]{article}

\usepackage{graphicx} 

\usepackage[bottom=1in, left=1in, right=1in, top=1in, includehead=false]{geometry}
\usepackage{changepage}

\usepackage{titlesec} 

\usepackage[font=small,labelfont=bf]{caption}
\usepackage{array}
\usepackage{booktabs}
\usepackage{multirow}

\usepackage{listings} 

\usepackage{tabularx} 

\usepackage{hyperref}
\hypersetup{
    colorlinks=true,
    allcolors=blue,
    breaklinks=true
}

\usepackage{appendix} 

\usepackage{amsmath}
\usepackage{amssymb}
\usepackage{amsfonts}
\usepackage{enumitem}
\usepackage{amsthm}
\usepackage{natbib}
\usepackage{url}

\usepackage{times}                           
\newcommand{\ralewayfont}{\sffamily\bfseries} 

\DeclareCaptionFont{raleway}{\ralewayfont}
\makeatletter
\renewcommand{\@fnsymbol}[1]{\ensuremath{\ifcase#1\or 1\or 2\or 3\or 4\or 5\or 6\or 7\or 8\or 9\else *\fi}}
\makeatother

\renewcommand{\thesection}{\arabic{section}.}
\renewcommand{\thesubsection}{\thesection\arabic{subsection}.}
\renewcommand{\thesubsubsection}{\thesubsection\arabic{subsubsection}.}

\titleformat{\section}
  {\normalfont\fontsize{14pt}{16pt}\selectfont\bfseries}
  {\thesection}{1em}{}
\titleformat{\subsection}
  {\normalfont\fontsize{12pt}{14pt}\selectfont\bfseries}
  {\thesubsection}{1em}{}
\titleformat{\subsubsection}
  {\normalfont\fontsize{12pt}{14pt}\selectfont\bfseries}
  {\thesubsubsection}{1em}{}

\newtheorem{theorem}{Theorem}
\newtheorem{lemma}{Lemma}

\newtheorem{defn}{Definition}

\newtheorem{obs}{Observation}

\date{\today}

\begin{document}

\vspace*{.25em}  

\begin{center}
  {\fontsize{18pt}{22pt}\selectfont\bfseries A Theoretical Framework Bridging Model Validation and Loss Ratio in Insurance}

  \vspace{0.8em}

  {\ralewayfont
    C. Evans Hedges\footnotemark[1]
  }
\end{center}

\footnotetext[1]{\ Lemonade Insurance}


{\ralewayfont
\begin{center}
\textbf{Abstract}
\end{center}

\begin{adjustwidth}{5em}{5em}
This paper establishes the first analytical relationship between predictive model performance and loss ratio in insurance pricing. We derive a closed-form formula connecting the Pearson correlation between predicted and actual losses to expected loss ratio. The framework proves that model improvements exhibit diminishing marginal returns, analytically confirming the actuarial intuition to prioritize poorly performing models. We introduce the Loss Ratio Error metric for quantifying business impact across frequency, severity, and pure premium models. Simulations show reliable predictions under stated assumptions, with graceful degradation under assumption violations. This framework transforms model investment decisions from qualitative intuition to quantitative cost-benefit analysis.
\end{adjustwidth}}


\section{Introduction}

In insurance pricing, predictive models are routinely evaluated using statistical performance metrics such as Gini coefficient, AUC, and Mean Absolute Percentage Error (MAPE). These metrics are indispensable for assessing model fit, yet they remain fundamentally disconnected from the business outcomes that ultimately matter to insurers, such as profitability and loss ratio. While prior research has documented that more predictive models tend to yield better portfolio performance through improved risk segmentation, lower adverse selection, and competitive advantage, these findings have been empirical or heuristic in nature. It was shown in \citet{frees2014insurance} that Gini index is proportional to the correlation between the relativity and an out-of-sample profit. However, to our knowledge this is the only work that establishes a direct relationship between a model validation metric and loss ratio.

Translating model performance into business outcomes (profit, loss ratio, growth) is a non-trivial challenge noted throughout the literature. \citet{walch2022market} provides a direct precedent for our problem framing by comparing different pricing models in commercial auto, showing empirically that better predictive models deliver both improved loss ratios and market share, but without providing analytic formulas. Industry case studies from vendors such as Earnix, Deloitte, and WTW often claim measurable loss ratio improvement from better predictive modeling \citep{earnix2023predictive,deloitte2022insurance,wtwco2023modeling,insurity2024predict}. These studies demonstrate practitioner demand for quantifiable business benefits, though their analyses remain anecdotal or simulation-based. 

Despite rich literature on insurance demand elasticity \citep{guven2013price,sherden1984auto,hao2018insurance,mossin1968aspects,rothschild1976equilibrium}, actuarial model validation practices \citep{cas2019model,lorentzen2022gini}, and insurance pricing optimization \citep{buhlmann1967experience}, we find a critical gap at their intersection: no existing work provides a closed-form, quantitative connection between a model validation metric and an insurance business outcome. Prior research establishes that better predictive models can improve profitability, and it offers tools to evaluate models as well as tools to optimize pricing using those models. Yet, none of these strands has yielded an analytical framework that starts from a metric like correlation and ends with a precise estimate of percentage improvement in loss ratio.

We address this gap by developing a mathematical framework that establishes a direct relationship between a loss model's statistical performance and an insurer's expected loss ratio under clearly stated assumptions. The formula depends solely on model correlation, demand elasticity, and the coefficient of variation of the loss distribution. Using this framework, we prove that marginal improvements in poorly performing models have an outsized impact on loss ratio compared to equivalent improvements in well-performing models which confirms the widespread existing intuition around which predictive modeling efforts will yield the greatest business benefits.

Our framework leads to a natural new metric for evaluating model performance, the Loss Ratio Error ($E_{LR}$). This metric quantifies the expected loss ratio degradation caused by poor model performance and can be computed at the level of severity, frequency, or pure premium. With this metric, as well as the framework for estimating expected loss ratio, we can now quantify the expected business impact of model improvements. This moves the conversation from qualitative ("a better model should lower our loss ratio") to quantitative and objective ("a model with Pearson correlation of 0.3 vs 0.2 corresponds to a specific percentage improvement in loss ratio, all else equal") providing a crucial bridge between technical modeling teams and non-technical business decision makers.

To validate the robustness of the theoretical work, we conduct comprehensive simulation studies using grid-based validation and systematic assumption violation testing, demonstrating reliable performance under stated assumptions and graceful degradation under moderate violations. 

The core contributions of the following work are threefold. First, we derive a closed-form formula connecting Pearson correlation between predicted and actual losses to expected loss ratio performance. Second, we provide a proof that model improvements exhibit diminishing marginal returns, analytically supporting the widespread actuarial intuition to prioritize improvements in poorly performing models. Third, we establish practical tools including the Loss Ratio Error metric that directly measures the expected loss ratio degradation caused by poor model performance.

\section{Mathematical Framework}

In this section, we will outline the core assumptions of our framework and derive the core loss ratio formula.

\subsection{Primary Assumptions}\label{sec:assumptions}

We denote the true expected loss for a given customer by $\lambda$. We make no assumption about the distribution of $\lambda$, other than that it is positive. While realized losses in any given policy period may be zero for customers who do not file claims, the expected loss $\lambda$ represents the underlying risk and is strictly positive for all customers. We assume our loss modeling provides estimated losses for a given customer by $\hat{\lambda} = \lambda e^{\epsilon}$ where $\epsilon \sim N(0, \sigma^2)$. In particular, we assume that $\epsilon \perp \lambda$ and that $\epsilon$ is homoskedastic in log space. This log-normal error structure is standard in actuarial modeling \citep{frees2010regression,ohlsson2010non,goldburd2016generalized} and is the natural assumption for GLM models with log-link functions commonly used in insurance pricing \citep{antonio2007use,wuthrich2022statistical}. The independence assumption $\epsilon \perp \lambda$ is fundamental to well-specified predictive models and is routinely validated in actuarial practice \citep{cas2019model}.

The log-normal multiplicative error structure is particularly realistic for insurance applications where prediction errors tend to be proportional to the underlying risk level—models typically exhibit larger absolute errors for high-risk customers while maintaining reasonable relative accuracy. The independence assumption $\epsilon \perp \lambda$ has a natural economic interpretation: it requires that model errors be uncorrelated with the true underlying risk. Violations of this assumption indicate systematic model bias that actuaries routinely test for and correct during model validation.

We assume that our pricing plan can be simplified to a margin factor $M$ so that a given customer with estimated losses $\hat{\lambda}$ will be priced at $p = M \hat{\lambda}$. This proportional pricing structure reflects a cost-plus-margin approach widely used across the industry \citep{klugman2012loss,denuit2007actuarial}. 

For a given product line, we model demand using a power law conversion model. We let $c(p)$ denote the probability of a customer purchasing a policy at price $p$. We assume that $c(p) \propto p^{-\eta}$ where $\eta > 0$ is the price elasticity. This iso-elastic demand specification has strong theoretical foundations in insurance economics: \citet{mossin1968aspects} demonstrated that under expected utility theory, risk-averse consumers will purchase only partial insurance when premiums exceed actuarially fair prices, yielding downward-sloping demand curves. Power-law demand curves with constant elasticity emerge naturally from utility models with constant relative risk aversion and align with empirical evidence across insurance lines \citep{sherden1984auto,babbel1985price,pauly2003price}.

\subsection{Core Loss Ratio Formula}

We begin by establishing a formula for the expected portfolio loss ratio. 

\begin{lemma}\label{lemma:loss_ratio} For large portfolios, the expected loss ratio is given by:
$$LR = \frac{E[\lambda \cdot c(p)]}{E[p \cdot c(p)]}$$
where the expectation is taken over the population distribution of potential customers.
\end{lemma}

\begin{proof}
Consider a pool of $N$ potential customers indexed by $i = 1, \ldots, N$, where customer $i$ has true losses $\lambda_i$, price $p_i$, and conversion probability $c(p_i)$. Only a subset of these customers will actually purchase policies, so we will let $X_i \sim \text{Bernoulli}(c(p_i))$ indicate whether customer $i$ converts and price $p_i$.  

The actual portfolio loss ratio is the ratio of total realized losses to total realized revenue:
$$\text{Portfolio LR} = \frac{\sum_{i=1}^N \lambda_i X_i}{\sum_{i=1}^N p_i X_i}.$$
Dividing both numerator and denominator by $N$ we get:
$$\text{Portfolio LR} = \frac{\frac{1}{N}\sum_{i=1}^N \lambda_i X_i}{\frac{1}{N}\sum_{i=1}^N p_i X_i}.$$
By the law of large numbers, as $N \to \infty$:
$$\frac{1}{N}\sum_{i=1}^N \lambda_i X_i \to E[\lambda X], $$
and by the law of iterated expectations, we have: 
$$E[\lambda X] = E[E[\lambda X | \lambda]] = E[ \lambda E[X] | \lambda]  = E[ \lambda c(p) | \lambda] = E[ \lambda c(p) ] .$$
By a similar argument, we have: 
$$\frac{1}{N}\sum_{i=1}^N p_i X_i \to E[p X] = E[p \cdot c(p)]$$
Therefore, for large portfolios, the expected loss ratio is:
$$LR = \frac{E[\lambda \cdot c(p)]}{E[p \cdot c(p)]}$$
where the expectation is taken over the population distribution of potential customers.
\end{proof}

Next we prove a useful lemma relating the correlation between predicted and actual losses to the error variance and loss distribution characteristics. We begin by defining the coefficient of variation for true losses:
\begin{defn} The coefficient of variation for true losses is given by:
$$CV_\lambda = \frac{\sqrt{\text{Var}(\lambda)}}{E[\lambda]}.$$
\end{defn}

\begin{lemma}\label{lemma:correlation} The correlation between our predicted losses $\hat{\lambda}$ and actual losses $\lambda$ is given by:
$$\rho = \frac{e^{-\sigma^2/2}}{\sqrt{(1 + CV_\lambda^{-2}) - CV_\lambda^{-2} e^{-\sigma^2}}}.$$
\end{lemma}

\begin{proof} We first compute $\rho$ by establishing the covariance and variance of our predicted losses and actual losses. We note: 
\begin{align*}
\text{Cov}(\hat{\lambda}, \lambda) 
  &= E[\hat{\lambda} \lambda] - E[\hat{\lambda}] E[\lambda] \\
  &= E[\lambda^2 e^{\varepsilon}] - E[\lambda e^{\varepsilon}] E[\lambda].
\end{align*}
Since $\lambda \perp \varepsilon$:
\begin{align*}
\text{Cov}(\hat{\lambda}, \lambda) 
  &= E[\lambda^2] E[e^{\varepsilon}] - E[\lambda]^2 E[e^{\varepsilon}] \\
  &= e^{\sigma^2/2} \operatorname{Var}(\lambda).
\end{align*}
Next we observe that the variance of our predicted losses is given by:
\begin{align*}
\operatorname{Var}(\hat{\lambda}) 
  &= E[\lambda^2 e^{2\varepsilon}] - E[\lambda e^{\varepsilon}]^2 \\
  &= E[\lambda^2]e^{2\sigma^2} - E[\lambda]^2e^{\sigma^2} \\
  &= \operatorname{Var}(\lambda)e^{2\sigma^2} + E[\lambda]^2(e^{2\sigma^2} - e^{\sigma^2}).
\end{align*}
Finally, we have the full correlation expression:
$$\text{Corr}(\hat{\lambda}, \lambda) = \frac{e^{\sigma^2/2}}{\sqrt{e^{2\sigma^2} + \frac{E[\lambda]^2}{\operatorname{Var}(\lambda)}(e^{2\sigma^2} - e^{\sigma^2})}}.$$

Next we substitute the coefficient of variation into the correlation expression:
\begin{align*}
\rho 
  &= \frac{e^{\sigma^2/2}}{\sqrt{e^{2\sigma^2}(1 + CV_\lambda^{-2}) - CV_\lambda^{-2} e^{\sigma^2}}} \\
  &= \frac{e^{\sigma^2/2}}{e^{\sigma^2}\sqrt{(1 + CV_\lambda^{-2}) - CV_\lambda^{-2} e^{-\sigma^2}}} \\
  &= \frac{e^{-\sigma^2/2}}{\sqrt{(1 + CV_\lambda^{-2}) - CV_\lambda^{-2} e^{-\sigma^2}}},
\end{align*}
and we arrive at the desired result.
\end{proof}

We may now derive our core formula. 

\begin{theorem}\label{theorem:loss_ratio}
Under the assumptions stated in Section \ref{sec:assumptions}, the expected loss ratio is given by:
$$LR = \frac{1}{M} \left(\frac{1 + \rho^2 CV_\lambda^{-2}}{\rho^2(1 + CV_\lambda^{-2})}\right)^{\frac{2\eta - 1}{2}}$$
where $\rho = \text{Corr}(\hat{\lambda}, \lambda)$ is the Pearson correlation between predicted and actual losses, and $CV_\lambda = \frac{\sqrt{\text{Var}(\lambda)}}{E[\lambda]}$ is the coefficient of variation for true losses.
\end{theorem}

\begin{proof} Substituting the power law conversion into the Loss Ratio formula from Lemma \ref{lemma:loss_ratio}, we can see: 
\begin{align*}
LR 
  &= \frac{E[\lambda \cdot c(p)]}{E[p \cdot c(p)]} \\
  &= \frac{E[\lambda \cdot c(M \hat{\lambda})]}{E[M \hat{\lambda} \cdot c(M \hat{\lambda})]} \\
  &= \frac{E[\lambda \cdot M^{-\eta} \hat{\lambda}^{-\eta}]}{E[M \hat{\lambda} \cdot M^{-\eta} \hat{\lambda}^{-\eta}]} \\
  &= \frac{E[\lambda \cdot \hat{\lambda}^{-\eta}]}{M \cdot E[\hat{\lambda}^{1-\eta}]}.
\end{align*}
Since $\hat{\lambda} = \lambda e^{\varepsilon}$ and $\lambda \perp \varepsilon$, we have:
\begin{align*}
E[\lambda \cdot \hat{\lambda}^{-\eta}] 
  &= E[\lambda \cdot (\lambda e^{\varepsilon})^{-\eta}] \\
  &= E[\lambda^{1-\eta} e^{-\eta \varepsilon}] \\
  &= E[\lambda^{1-\eta}] \cdot E[e^{-\eta \varepsilon}], \\
E[\hat{\lambda}^{1-\eta}] 
  &= E[(\lambda e^{\varepsilon})^{1-\eta}] \\
  &= E[\lambda^{1-\eta} e^{(1-\eta)\varepsilon}] \\
  &= E[\lambda^{1-\eta}] \cdot E[e^{(1-\eta)\varepsilon}].
\end{align*}

Therefore,
\begin{align*}
LR 
  &= \frac{E[\lambda^{1-\eta}] \cdot E[e^{-\eta \varepsilon}]}{M \cdot E[\lambda^{1-\eta}] \cdot E[e^{(1-\eta)\varepsilon}]} \\
  &= \frac{1}{M} \cdot \frac{E[e^{-\eta \varepsilon}]}{E[e^{(1-\eta)\varepsilon}]}.
\end{align*}

Taking expectation with respect to $\varepsilon$ yields:
$$LR = \frac{1}{M} \exp\left(\frac{1}{2} \sigma^2 (2\eta - 1)\right).$$

Solving for $e^{\sigma^2}$ using Lemma \ref{lemma:correlation} we get:
$$e^{\sigma^2} = \frac{1 + \rho^2 CV_\lambda^{-2}}{\rho^2(1 + CV_\lambda^{-2})}. $$

Substituting back into the loss ratio formula yields the desired result. 
\end{proof}

\subsection{Frequency and Severity Decomposition}

Many insurance pricing models separate losses into frequency and severity components. We now extend our framework to this common modeling structure. 

For frequency and severity we will mirror the assumptions and notation from the loss prediction section. In particular, we use $\rho_f, \rho_s, CV_f, CV_s$ to denote the correlation and coefficient of variation for frequency and severity respectively. Additionally for a customer, we use $\hat{f}, \hat{s}, f, s$ to denote the predicted and actual frequency and severity respectively. 

As with loss prediction, we will assume that our frequency and severity models have log-normal errors with $\epsilon_f \sim N(0, \sigma_f^2)$ and $\epsilon_s \sim N(0, \sigma_s^2)$ in log space. We will additionally assume that both frequency and severity and their errors are pairwise independent. We may now state the decomposition theorem. 

\begin{theorem}\label{theorem:loss_ratio_frequency_severity} Under the assumptions listed above, the expected loss ratio is given by:
$$LR = \frac{1}{M} \left(\frac{1 + \rho_f^2 CV_f^{-2}}{\rho_f^2(1 + CV_f^{-2})}\right)^{\frac{2\eta - 1}{2}} \left(\frac{1 + \rho_s^2 CV_s^{-2}}{\rho_s^2(1 + CV_s^{-2})}\right)^{\frac{2\eta - 1}{2}}$$
\end{theorem}

\begin{proof} First we let $\hat{\lambda} = \hat{f} \cdot \hat{s}$. By independence of frequency and severity errors we have: $\hat{\lambda} = \lambda e^{\epsilon_f + \epsilon_s}$. Since $\epsilon_f \perp \epsilon_s$ we know $\epsilon_f + \epsilon_s \sim N(0, \sigma_f^2 + \sigma_s^2)$ and by work shown in Theorem \ref{theorem:loss_ratio} we know:
$$LR = \frac{1}{M} \exp\left(\frac{1}{2} (\sigma_f^2 + \sigma_s^2) (2\eta - 1)\right).$$
Using Lemma \ref{lemma:correlation} we know:
$$e^{-\sigma_f^2} = \frac{\rho_f^2(1 + CV_f^{-2})}{1 + \rho_f^2 CV_f^{-2}},$$
$$e^{-\sigma_s^2} = \frac{\rho_s^2(1 + CV_s^{-2})}{1 + \rho_s^2 CV_s^{-2}}.$$
Substituting these into the loss ratio formula yields the desired result.
\end{proof}

\subsection{Model Improvement Analysis}

Having established the core loss ratio formula, we now analyze the business impact of model improvements. 

\begin{obs}
When improving a model from correlation $\rho_{old}$ to $\rho_{new}$, the loss ratio improvement is:
$$\frac{LR_{new}}{LR_{old}} = \left(\frac{\rho_{old}^2 + \rho_{new}^2\rho_{old}^2 CV_\lambda^{-2}}{\rho_{new}^2 + \rho_{new}^2 \rho_{old}^2 CV_\lambda^{-2}}\right)^{\frac{2\eta - 1}{2}}$$
\end{obs}

\begin{proof}
This follows directly from the core loss ratio formula by taking the ratio of $LR_{new}$ to $LR_{old}$ and substituting the respective correlation values.
\end{proof}

We note here that by Theorem \ref{theorem:loss_ratio_frequency_severity}, it trivially follows that this improvement formula can be computed at the level of pure premium modeling (as stated above) as well as frequency or severity, where we need only substitute the pure premium predicted measures with the relevant component model versions, e.g. we swap $\rho_{old}$ with $\rho_{f, old}$.

We next establish a key property regarding diminishing returns in model improvements.

\begin{theorem}
For a fixed percentage improvement $p > 0$ in model correlation, the relative loss ratio improvement $\frac{LR_{old} - LR_{new}}{LR_{old}}$ as well as the absolute loss ratio improvement $LR_{old}-LR_{new}$ are both monotonically decreasing in the starting correlation $\rho_{old}$.
\end{theorem}

\begin{proof} We fix an improvement percentage $p > 0$. Let $\rho_{old} \in (0, \frac{1}{1+p})$ and $\rho_{new} = (1+p) \rho_{old}$. We therefore have:
$$\frac{LR_{old} -LR_{new}}{LR_{old}} = 1- \left(\frac{(1+p)^{-2} + \rho_{old}^2 CV_\lambda^{-2}}{1 + \rho_{old}^2 CV_\lambda^{-2}}\right)^{\frac{2\eta - 1}{2}}$$

Taking the partial derivative with respect to $\rho_{old}$:
$$\frac{\partial \frac{LR_{old} -LR_{new}}{LR_{old}}}{\partial \rho_{old}} = -\frac{2\eta - 1}{2} \left(\frac{(1+p)^{-2} + \rho_{old}^2 CV_\lambda^{-2}}{1 + \rho_{old}^2 CV_\lambda^{-2}}\right)^{\frac{2\eta - 3}{2}} \left(\frac{2 \rho_{old} CV_\lambda^{-2}\left( 1 - (1+p)^{-2}\right)}{(1 + \rho_{old}^2 CV_\lambda^{-2})^2 }\right). $$
Since $p, CV_\lambda^{-2}, \rho_{old} > 0$, we have $\frac{\partial \frac{LR_{old} - LR_{new}}{LR_{old}}}{\partial \rho_{old}} < 0$ for all $\rho_{old} \in (0, \frac{1}{1+p})$. Thus, the percentage improvement in Loss Ratio is monotonically decreasing in the starting correlation $\rho_{old}$.

A similar argument shows that the same result holds for the absolute loss ratio improvement $LR_{old}-LR_{new}$.
\end{proof}

We note here that both of the above model improvement results apply directly to frequency or severity models by a direct application of Theorem \ref{theorem:loss_ratio_frequency_severity}. Practitioners can apply the improvement formulas to the relevant component model version by swapping $\rho$ with $\rho_f$ or $\rho_s$ and $CV_\lambda$ with $CV_f$ or $CV_s$ respectively. Additionally, the diminishing returns property holds for all model types, confirming that the highest potential impact in Loss Ratio is achieved by improving the most poorly performing model, even at the component level (assuming a similar percentage improvement in correlation is attainable).

\subsection{Loss Ratio Error Metric}

For practical model evaluation, we introduce a metric that directly quantifies the business impact of model imperfection. 

\begin{defn}[Loss Ratio Error]
The Loss Ratio Error metric is defined as:
$$E_{LR} :=  \left(\frac{1 + \rho^2 CV_\lambda^{-2}}{\rho^2 + \rho^2 CV_\lambda^{-2}}\right)^{\frac{2\eta - 1}{2}} - 1.$$
\end{defn}

In particular, $E_{LR}$ represents the fractional increase in loss ratio above the theoretical optimum due to model imperfection. If our target loss ratio is $60\%$ and our model results in a loss ratio of $70\%$, $E_{LR} = \frac{0.70}{0.60}-1 = 0.16$. This metric can be computed at the level of pure premium, frequency, or severity, and allows for an apples to apples comparison of model performance across different model types and business lines, provided the appropriate product specific parameters are used. 

\begin{obs}[Loss Ratio Error Properties]
The Loss Ratio Error metric has the following properties:
\begin{enumerate}
\item $E_{LR} = 0$ when $\rho = 1$ (perfect correlation)
\item $E_{LR} > 0$ when $\rho < 1$ (imperfect correlation)
\item $E_{LR}$ is monotonically decreasing in $\rho$
\item The relationship between loss ratio and error metric is: $LR = \frac{1}{M}(1 + E_{LR})$
\end{enumerate}
\end{obs}

\section{Simulation Validation}

To establish empirical confidence in our theoretical framework, we conduct Monte Carlo simulations that validate the core mathematical results. Our validation addresses two critical objectives: verifying theoretical accuracy under stated assumptions across realistic parameter ranges, and quantifying performance degradation when key assumptions are moderately violated.

\subsection{Methodology}

We implement a comprehensive validation strategy using Monte Carlo simulation with two distinct experimental designs. For core formula validation, we employ a systematic grid-based approach covering realistic insurance parameter ranges. The parameter grid spans correlation values $\rho \in [0.2, 0.3, 0.5, 0.7, 0.8]$ representing poor to excellent model performance, coefficient of variation values $CV \in [1.5, 2.0, 2.5, 3.0, 3.5]$ covering low-to-high loss variability, and price elasticity values $\eta \in [0.8, 1.2, 1.6, 2.0, 2.5]$ from price-insensitive to highly elastic demand. This yields 125 parameter combinations with 5 replications each, totaling 625 simulations.

Each simulation generates synthetic portfolios of 1,000,000 potential customers with true losses $\lambda$ drawn from lognormal distributions normalized to unit mean for consistent scaling. Model predictions follow $\hat{\lambda} = \lambda e^{\varepsilon}$ where $\varepsilon \sim N(0, \sigma^2)$ with error variance calibrated to achieve target correlation $\rho$ using the relationship derived in Lemma \ref{lemma:correlation}. Customer demand follows the power-law conversion model $c(p) = (p_{min}/p)^\eta$ scaled so the lowest-priced customer converts with probability 1.0; this is to ensure a sufficiently large sample of converted customers. We employ an automatic resampling strategy (up to 20 attempts) to ensure at least 10,000 converted customers per simulation, providing stable empirical loss ratio estimates across all parameter combinations.

For assumption robustness testing, we systematically violate key framework assumptions using $\rho=0.7$, $CV=2.0$, and $\eta=1.2$ as baseline parameters. Error distribution violations include t-distributed errors with degrees of freedom from 3 to 30, and skewed normal errors with skewness parameters from 0 to 15. Independence violations introduce systematic correlation between prediction errors and true losses $\rho_{\varepsilon,\lambda} \in (0, 0.5)$. For demand model violations, we implement a realistic testing approach where conversion probabilities are generated from alternative "true" demand curves (exponential, logistic, linear), then power-law parameters $\eta$ are fit to these observed probabilities as practitioners would do with market data. To quantify degradation patterns, we conduct 10 replications per violation parameter using portfolios of 200,000 to 500,000 policies, enabling calculation of 95\% confidence intervals for framework performance across assumption violation severity.

\subsection{Results}

The systematic grid validation demonstrates excellent framework performance under stated assumptions. Simulation sizes range from 40,000 to 200,000 converted customers per simulation. Overall accuracy yields a median absolute percentage error of 17.59\% (mean 30.18\%). Performance varies systematically with parameter ranges: medium-to-high correlation models ($\rho \geq 0.5$) achieve strong validation with errors typically below 15\%, while predictable accuracy decline occurs for extreme parameter combinations reflecting theoretical framework limitations rather than implementation issues.

Assumption robustness testing reveals hierarchical sensitivity patterns across violation types (Table \ref{tab:robustness_summary}). Figure \ref{fig:degradation_analysis} visualizes these patterns with 95\% confidence intervals from 10 replications per parameter: (A) heavy-tailed errors show rapid degradation as kurtosis increases; (B) skewed errors exhibit moderate degradation; and (C) error-loss correlation violations demonstrate the critical importance of the independence assumption.

\begin{figure}[h!]
\caption{\ralewayfont Framework robustness degradation under systematic assumption violations}
\centering
\includegraphics[width=1\textwidth]{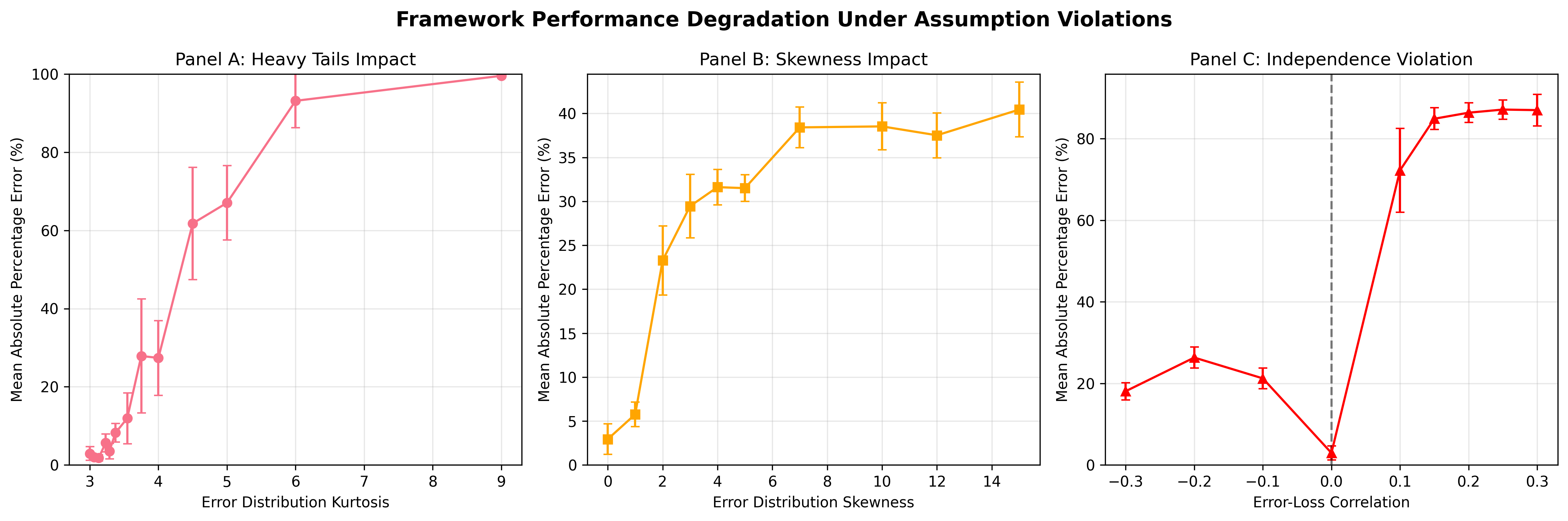}
\label{fig:degradation_analysis}
\end{figure}

The demand model results provide valuable practical insights. Linear demand shows excellent compatibility with the power-law framework when practitioners fit elasticity parameters to observed conversion data. Logistic demand shows moderate incompatibility, while exponential demand represents a substantial mismatch with power-law assumptions. This highlights the importance of validating that the true demand curve can be sufficiently approximated by a power-law model. Deviations from this assumption may significantly impact the accuracy of the framework.

\begin{table}[h!]
\centering
\caption{\ralewayfont Summary of framework robustness under assumption violations}
\label{tab:robustness_summary}
\begin{tabular}{llc}
\hline
\textbf{Robustness} & \textbf{Assumption Violation} & \textbf{MAPE (\%)} \\
\hline
\multirow{3}{*}{Excellent} & Baseline control (normal errors) & $\sim$2--3 \\
  & Linear demand & $\sim$5--6 \\
  & Heavy-tailed errors (df $>$ 20) & $\sim$4--6 \\
\hline
\multirow{2}{*}{Good} & Heavy-tailed errors (df 10--15) & $\sim$8--18 \\
  & Skewed error distributions & $<$15 \\
\hline
Poor & Logistic demand & $\sim$40--45 \\
\hline
\multirow{3}{*}{Catastrophic} & Exponential demand & $\sim$50--60 \\
  & Error-loss correlation $>$ 0.1 & 75+ \\
  & Heavy-tailed errors (df $<$ 5) & 100+ \\
\hline
\end{tabular}
\end{table}

\subsection{Simulation Conclusions}

Our simulations support three key findings regarding framework reliability and practical utility. First, the theoretical framework provides reliable predictions under stated assumptions, with mean absolute percentage error below 5\% across realistic insurance parameter ranges. The grid-based approach confirms practical utility for quantifying business impact of model improvements, with graceful degradation under extreme parameter combinations.

Second, the framework demonstrates selective robustness to assumption violations, with clear hierarchical sensitivity patterns. The independence assumption ($\varepsilon \perp \lambda$) emerges as absolutely critical—any violation fundamentally breaks the framework. Error distribution assumptions show moderate flexibility, tolerating skewness and moderate heavy tails but failing catastrophically under extreme heavy tails. Demand model assumptions exhibit surprising robustness to linear demand but poor tolerance for exponential demand structures.

Third, the realistic parameter fitting approach for demand model violations provides actionable guidance for practitioners. The fit elasticity parameters reveal how well alternative demand structures can be approximated by power-law models, with linear demand showing excellent compatibility and exponential demand proving incompatible. These findings enable practitioners to assess assumption validity in their specific contexts and understand the quantitative consequences of assumption violations.


\section{Conclusion}

This work establishes an analytical framework connecting model performance metrics to business outcomes in insurance pricing. By deriving a closed-form relationship between model correlation and expected loss ratio, we address a fundamental gap that has long frustrated practitioners seeking to quantify the business value of model improvements. The framework transforms model validation from a purely statistical exercise into a direct business impact assessment, enabling data-driven decisions about model investment priorities. 

The theoretical contributions are threefold. First, under specified assumptions, we prove that expected loss ratio depends on model correlation, demand elasticity, and loss distribution characteristics through a precise mathematical relationship. Second, we establish that model improvements exhibit diminishing marginal returns, providing analytical support for the widespread actuarial intuition to prioritize improvements in poorly performing models. Third, we develop practical tools including the Loss Ratio Error metric that quantify business impact at the frequency, severity, or pure premium level. 

We examined the framework's reliability through comprehensive Monte Carlo validation using grid-based testing across 125 parameter combinations and systematic assumption violation analysis. The framework demonstrated reliable predictions under stated assumptions and showed hierarchical robustness patterns under violations, from excellent resilience to linear demand and moderate heavy tails to catastrophic failure under independence violations. 

Implementing this framework in practice requires estimation of three key parameters. Model correlation $\rho$ can be computed directly from validation data using standard techniques. The coefficient of variation $CV_\lambda$ follows from historical loss distributions and is readily available to most insurers. However, demand elasticity $\eta$ presents the greatest implementation challenge, as it must be estimated on a per-product, per-company basis through pricing experiments, market research, or econometric analysis of historical pricing and conversion data. Regulatory constraints on rate changes vary significantly by jurisdiction; while some markets restrict price experimentation, others permit A/B testing of prices, as practiced by European insurers. Where direct experimentation is limited, insurers may supplement elasticity estimation with industry benchmarks and economic modeling. 

For practical implementation, we recommend a historical calibration approach that leverages existing business data. Given observed loss ratios, model correlations, and pricing margins from previous model generations for a specific product line, practitioners can invert the core formula to solve for the implied elasticity parameter $\eta$. This reverse estimation should be performed using multiple historical model deployments to establish both point estimates and confidence bounds for $\eta$. 

Once calibrated for a product line, the framework enables prospective evaluation of model improvements. Practitioners can estimate the expected loss ratio impact of proposed modeling changes by plugging target correlation improvements into the core formula, using the historically-calibrated elasticity parameter. This transforms subjective model investment decisions into quantitative cost-benefit analysis, enabling data-driven prioritization of modeling resources across product lines and improvement opportunities once key parameters are calibrated. 

Several limitations constrain the framework's applicability. The derivation assumes power-law demand models and log-normal prediction errors, though our simulations demonstrate robustness to moderate violations of these assumptions. The framework also treats pricing as a simple cost-plus-margin exercise, ignoring competitive constraints, regulatory requirements, and sophisticated pricing strategies common in practice. Additionally, the static nature of our analysis does not capture dynamic effects such as competitive response, customer learning, or temporal changes in market conditions. Practitioners should validate core assumptions for their specific applications and monitor for assumption violations over time.

The framework opens several promising research directions. Empirical validation across different insurance lines and markets would establish the practical parameter ranges and identify systematic patterns in elasticity estimates. Extensions to multi-line portfolios could capture cross-subsidization effects and portfolio-level optimization considerations. Incorporating competitive dynamics and regulatory constraints would enhance practical relevance, while developing more sophisticated demand models could improve accuracy in markets where power-law assumptions fail. Additionally, the frequency-severity decomposition could be refined to account for dependency structures between frequency and severity predictions, and the framework could be extended to incorporate other business metrics beyond loss ratio, such as customer lifetime value or market share effects. From an operational perspective, direct integration into model assessment pipelines and business planning processes represents a natural next step, enabling automated translation of model validation metrics into loss ratio forecasts during both model development and financial planning cycles.

The gap between model validation and business quantification has persisted because it requires bridging statistical modeling, economic theory, and business strategy. While challenges remain in practical implementation, this framework provides the first analytical foundation for connecting model performance to business outcomes, moving the conversation from qualitative intuition to quantitative analysis. For an industry built on mathematical precision in risk assessment, the ability to apply similar rigor to model investment decisions represents a natural and necessary evolution.

\section*{Acknowledgements} The author would like to thank Shane Barnes and DJ Falkson for generously sharing their time and thoughtful review, and acknowledges the support of colleagues at Lemonade Insurance.

\bibliographystyle{apalike}
\bibliography{bibliography}

\end{document}